\pdfoutput=1 		

\documentclass{IET}

\usepackage{graphicx,color}	
\usepackage{epstopdf}	
\usepackage{amsfonts,amssymb,bm}

\usepackage{url}
\usepackage{cite}
\usepackage{units}

\usepackage{enumitem}
\setlist[enumerate,1]{label=(\roman*)}

\graphicspath{{.},{figures/}}
\newcommand{\deltac}{\delta^\text{\normalfont{ctrl}}}
\newcommand{\deltae}{\delta^\text{est}}

\newcommand{\Nsens}{N_\text{sen}}
\newcommand{\Nest}{N_\text{est}}
\newcommand{\Nall}{N}
\newcommand{\Ic}{I^\text{c}}

\newcommand{\transpose}{\text{T}}
\DeclareMathOperator*{\avg}{avg}
\DeclareMathOperator*{\diag}{diag}

\newcommand{\Ts}{T_\text{s}}

\newcommand{\field}[1]{\mathbb{#1}}
\newcommand{\R}{\field{R}}
\newcommand{\N}{\field{N}}

\newcommand{\E}{\mathbb{E}}

\newcommand{\ie}{i\/.\/e\/.\/~}
\newcommand{\eg}{e\/.\/g\/.\/~}
\newcommand{\cf}{cf\/.\/~}
\newcommand{\fig}{Fig\/.\/~}
\newcommand{\tab}{Table~}
\newcommand{\sect}{Section~}

\newcommand{\The}{Theorem~}
\newcommand{\Lem}{Lemma~}
\newcommand{\Pro}{Problem~}

\newcommand{\Assump}{Assumption~}

\newcommand{\etal}{\emph{et al.}}

\newtheorem{theorem}{Theorem}
\newtheorem{lemma}{Lemma}
\newtheorem{corollary}{Corollary}

\newtheorem{assum}{Assumption}
\newtheorem{remark}{Remark}
\newtheorem{problem}{Problem}

\usepackage{fancyhdr}

\newcommand{\mytitle}{\textbf{Accepted final version.}
This paper is a postprint of a paper submitted to and accepted for publication in IET Control Theory \& Applications
and is subject to Institution of Engineering and Technology Copyright. The copy of record is available at the IET Digital Library.
(doi: 10.1049/iet-cta.2016.1021)
}

\begin{document}

\title{Event-based State Estimation: An Emulation-based Approach}

\author[1*]{Sebastian Trimpe}
\affil{Autonomous Motion Department, Max Planck Institute for Intelligent Systems, Spemannstr.~38, 72076 T\"ubingen, Germany}
\affil[*]{strimpe@tuebingen.mpg.de}

\abstract{
An event-based state estimation approach for reducing communication in a networked control system is proposed. Multiple distributed sensor agents observe a dynamic process and sporadically transmit their measurements to estimator agents over a shared bus network.  Local event-triggering protocols ensure that data is transmitted only when necessary to meet a desired estimation accuracy.  The event-based design is shown to emulate the performance of a centralized state observer design up to guaranteed bounds, but with reduced communication.  The stability results for state estimation are extended to the distributed control system that results when the local estimates are used for feedback control.  Results from numerical simulations and hardware experiments illustrate the effectiveness of the proposed approach in reducing network communication.
%
}

\maketitle

\thispagestyle{fancy}	

\section{Introduction}
\label{sec:intro}
In almost all control systems today, data is processed and transferred between the system's components periodically.
%
While periodic system design is often convenient and well understood,
it involves an inherent limitation: data is processed and transmitted at  predetermined time instants, irrespective of the current state of the system or the information content of the data. That is, system resources are used regardless of whether there is any need for processing and communication or not. This becomes prohibitive when resources are scarce, such as in networked or cyber-physical systems, where multiple agents share a communication medium.
%

Owing to the limitations of traditional design methodologies for resource-constrained problems, aperiodic or event-based 
strategies have recently received a lot of attention \cite{Le11,Lu16}.
With event-based methods, data is transmitted or processed only when certain \emph{events} indicate that an update is required, for example, to meet some control or estimation specification.  Thus, resources are used \emph{only when required} and saved otherwise.

\begin{figure}[tb]
\centering
\includegraphics[width=0.7\columnwidth]{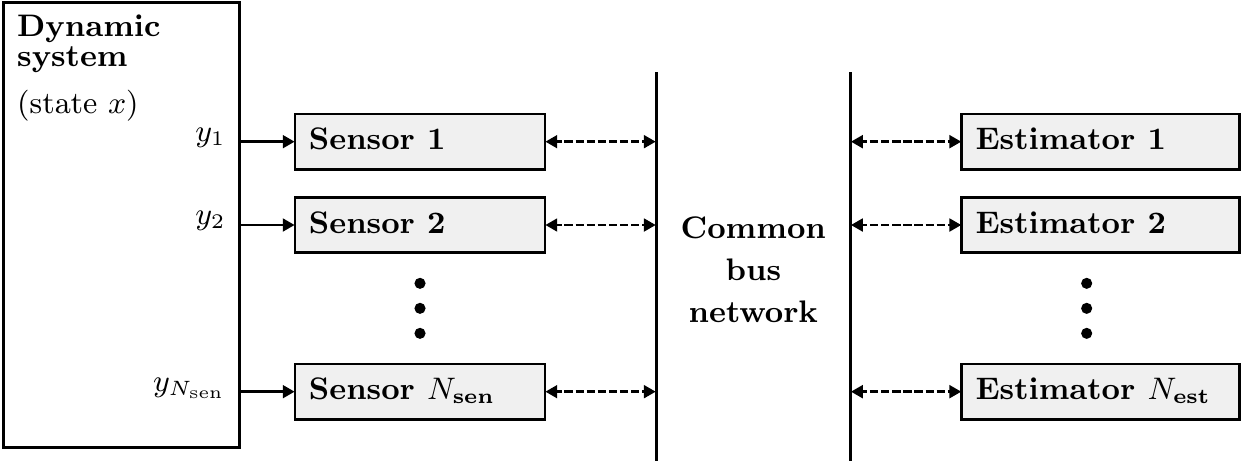}
\caption{Distributed state estimation problem.
Multiple distributed sensors make observations $y_i$ of a dynamic system and communicate to estimator nodes via a common bus network.   The development of an event-based scheme allowing all estimators to estimate the full system state $x$, but with limited inter-agent communication, is the objective of this article.  
}
\label{fig:networkArchitecture}
\end{figure}

In this article, a novel event-based scheme for distributed state estimation is proposed. We consider the system shown in \fig \ref{fig:networkArchitecture}, where multiple sensors observe a dynamic system and transmit data to estimator agents over a common bus.  Each estimator agent shall estimate the full state of the dynamic system, for example, for the purpose of monitoring or control.
In order to limit network traffic, 
local event triggers on each sensor ensure that updates are sent only when needed.  The common bus ensures that transmitted data reaches all agents in the network, which will allow for efficient triggering decisions and the availability of full state information on all agents.
%
%

The proposed approach for distributed event-based estimation emulates a classic 
discrete-time state observer
design up to guaranteed bounds, but with limited communication. \emph{Emulation-based design} is one common approach in event-based control literature (see \cite{Lu16}), where an event-based control system is designed so as to 
emulate the behavior of a given continuous or periodic control system.  However, to the best of the author's knowledge, emulation-based design has not been considered for state estimation.  
While the focus of this article is on state estimation, we also show stability of the event-based control system resulting when local estimates are used for feedback control. 

In particular, this articles makes the following main contributions:
\begin{enumerate}
\item First emulation-based design for distributed event-based state estimation replicating a centralized discrete-time linear observer.


\item Stability proofs for the resulting distributed and switching estimator dynamics under generic communication or computation imperfections (bounded disturbances).

\item Extension 
to distributed event-based control, where local estimates are used for feedback.

\item Experimental validation on an unstable networked control system.
\end{enumerate}


Preliminary results of those herein were presented in the conference papers \cite{Tr12,Tr14}; this article has been completely rewritten and new results added. 

\pagestyle{plain}

\subsection{Related work}

Early work on event-based state estimation (EBSE) concerned problems with a single sensor and estimator node
(see \cite{Le11} and references therein).
Typically, the optimal triggering strategies have
time-varying thresholds for finite-horizon problems, and constant thresholds for infinite-horizon problems, \cite[p.~340]{Le11}. Because long-term behavior (stability) is of primary interest herein,
we consider constant thresholds.

Different types of stationary triggering policies have been suggested in literature.  With the \emph{send-on-delta} (SoD) protocol \cite{Mi06}, transmissions are triggered based on the difference of the current and last-transmitted measurement.   
%
\emph{Innovation-based triggering} \cite{TrDAn11}
 places a threshold on the measurement innovation; that is, the difference of the current measurement and its prediction based on a process model.  
%
%
Wu \etal \cite{WuJiJoSh13} use the same trigger, but apply a transformation to decorrelate the innovation. Considering the variance of the innovation instead yields \emph{variance-based triggering} \cite{TrDAn14b}.
%
Marck and Sijs \cite{MaSi10} proposed \emph{relevant sampling}, where the relative entropy of prior and posterior state distribution is employed as a measure of information gain. 
%
We use innovation-based triggers herein, which have been shown to be 
effective for EBSE, \cite{TrCa15}.
Different estimation algorithms have been proposed for EBSE, with particular emphasis on how to (approximately) incorporate information contained in `negative' events (instants when no data is transmitted),  \cite{SiLa12,SiNoHa13,ShChSh15}.  
%
%
If one ignores the extra information from negative events in favor of a straightforward implementation, a time-varying Kalman filter (KF) can be used (\eg \cite{TrDAn11}).  Herein, we use the same structure as the standard KF, but with pre-computed switching gains, thus achieving the lowest computational complexity of all mentioned algorithms. 

To the best of the author's knowledge, \emph{distributed} EBSE with multiple sensor/estimator nodes and general coupled dynamics was first studied in \cite{TrDAn11}.  
 While Yook \etal \cite{YoTiSo02} had previously proposed the use of state estimators for the purpose of saving communication, they do not employ state estimation in the usual sense.  Instead of \emph{fusing} model-based predictions with incoming data, they \emph{reset} parts of the state vector.
%
%
%
%
Later results on distributed EBSE include \cite{BaBeCh12,ShChSh14c,ShChSh15,YaZhZhYa14,LiWaHeZh15,WeMoSiHaSh16}.  In contrast to the scenario herein, they consider either a centralized fusion node, or simpler SoD-type triggers, which are less effective for estimation, \cite{TrCa15}.
%
None of the mentioned references treats the problem of emulating a centralized observer design with a distributed and event-triggered implementation.

When the event-based state estimators are connected to state-feedback controllers (discussed in \sect \ref{sec:control}), this represents a \emph{distributed event-based control} system.  
Wang \etal \cite{WaLe11} and Mazo Jr \etal \cite{MaTa11} 
were among the first to discuss distributed or decentralized 
event-based control.
In contrast to these works, we neither assume perfect state measurements, nor a centralized controller as in  \cite{MaTa11}, nor have a restriction on the dynamic couplings \cite{WaLe11}, but we rely on a common bus network supporting all-to-all communication.
\subsection{Notation}
\label{sec:notation}
The terms \emph{state observer} and \emph{state estimator} are used synonymously in this article.
$\R$, $\N$, and $\N_N$ denote real numbers, positive integers, and the set $\{1, 2, \dots, N\}$, respectively.  
Where convenient, vectors are expressed as tuples $(v_1, v_2, \dots)$, where $v_i$ may be vectors themselves, with dimension and stacking clear from context.
For a vector $v$ and matrix $A$, $\|v\|$ denotes some vector H\"older norm \cite[p.~344]{Be05}, and $\|A\|$ the induced matrix norm.
%
For a 
sequence $v = \{v(0), v(1), \dots \}$, $\|v\|_\infty$ denotes the $\ell^\infty$ norm $\|v\|_\infty := \sup\nolimits_{k\geq 0} \, \|v(k)\|$.
%
%
%
For an estimate of $x(k)$ computed from measurement data until time $\ell \leq k$, we write $\hat{x}(k|\ell)$; and use $\hat{x}(k)=\hat{x}(k|k)$.  
A matrix is called stable if all its eigenvalues have magnitude strictly less than one.
Expectation is denoted by $\E[\cdot]$.

\section{Problem statement: distributed state estimation with reduced communication}
\label{sec:problemFormulation}
We introduce the considered networked dynamic system and state the estimation problem addressed in this article.

\subsection{Networked dynamic system}
We consider the networked estimation scenario in \fig \ref{fig:networkArchitecture}.  The dynamic system is described by linear discrete-time dynamics
\begin{align}
x(k) &= A x(k-1) + B  u(k-1) + v(k-1) \label{eq:system_x} \\
y(k) &= C  x(k) +w(k) \label{eq:system_y}
\end{align}
with sampling time $\Ts$, state $x(k) \in \R^n$, control input $u(k) \in \R^{q}$, measurement $y(k) \in \R^{p}$, disturbances $v(k) \in \R^n$, $w(k) \in \R^{p}$, and all matrices of corresponding dimensions.  We assume that $(A,B)$ is stabilizable and $(A,C)$ is detectable.  
No specific assumptions on the characteristics of the disturbances $v(k)$ and $w(k)$ are made; they can be random variables or deterministic disturbances. 

Each of the $\Nsens$ sensor agents (\cf \fig \ref{fig:networkArchitecture}) observes part of the dynamic process through
measurements $y_i(k) \in \R^{p_i}$, $i \in \N_{\Nsens}$.  The vector $y(k)$ thus represents the collective measurements of all $\Nsens$ agents,
\begin{align}
y(k) &= ( y_1(k), y_2(k), \dots, y_{\Nsens}(k) ) \\
y_i(k) &= C_i x(k) +w_i(k)  \qquad \forall \, i \in \N_{\Nsens}  \label{eq:system_yi}
\end{align}
with $C_i \in \R^{p_i \times n}$ and $w_i(k) \in \R^{p_i}$.
Agents can be heterogeneous with different types and dimensions of measurements, and no local observability assumption is made (\ie $(A,C_i)$ can be not detectable).

Each of the $\Nest$ estimator agents (\cf \fig \ref{fig:networkArchitecture}) shall reconstruct the full state 
for the purpose of, for example, having full information at different monitoring stations, distributed optimal decision making,
or local state-feedback control.
Overall, there are $\Nall = \Nsens + \Nest$ agents, and we use $i=1, \dots, \Nsens$ to index the sensor agents, and $i=1+\Nsens, \dots, \Nest+\Nsens$ for the estimator agents.

While the primary concern is the development of an event-based approach to the \emph{distributed state estimation} problem in \fig \ref{fig:networkArchitecture}, we shall also address \emph{distributed control} when the local estimates are used for feedback.
For this, we consider the control input decomposed as 
\begin{equation}
u(k) = ( u_1(k), u_2(k), \dots, u_{\Nest}(k) )
\label{eq:u_decomposed}
\end{equation}
with $u_i(k) \in \R^{q_i}$ the input computed on estimator agent $i + \Nsens$.

All agents are connected over a common-bus network; that is, if one agent communicates, all agents will receive the data.  
We assume that the network bandwidth is such that, in the worst case, all agents can communicate in one time step $\Ts$, and contention is resolved via low-level protocols.
Moreover, agents are assumed to be synchronized in time, and network communication is abstracted as instantaneous.

\begin{remark}
The common bus is a key component of the developed event-based approach.  
It will allow the agents to compute consistent estimates and use these for effective triggering decisions (while inconsistencies can still happen due to data loss  or delay).
Wired networks with a shared bus architecture 
such as Controller Area Network (CAN) or other 
fieldbus systems
are common in industry \cite{Th05}.
%
%
%
%
%
%
Recently, Ferrari \etal\ \cite{FeZiMoTh12} 
have proposed a common bus concept also for multi-hop low-power wireless networks.  
\end{remark}

\subsection{Reference design}
\label{sec:referenceDesign}
We assume that a centralized, discrete-time state estimator design is given, which we seek to emulate with the event-based design to be developed herein:
\begin{align}
\hat{x}_\text{c}(k|k-1) &= A  \hat{x}_\text{c}(k-1|k-1) + B  u(k-1) \label{eq:FCSE1} \\
\hat{x}_\text{c}(k|k) &= \hat{x}_\text{c}(k|k-1) + L  \big(y(k) - C \, \hat{x}_\text{c}(k|k-1) \big) 
\label{eq:FCSE2}
\end{align}
where the estimator gain $L \in \R^{n \times p}$ has been designed to achieve desired estimation performance, and the estimator is initialized with some $\hat{x}_\text{c}(0) = \hat{x}_\text{c}(0|0)$.
For example, \eqref{eq:FCSE1}, \eqref{eq:FCSE2} can be a Kalman filter representing the optimal Bayesian estimator for Gaussian noise,
or a Luenberger observer designed via pole placement 
to achieve a desired dynamic response.
At any rate, a reasonable observer design will ensure stable estimation error dynamics
\begin{align}
\epsilon_\text{c}(k) &:= x(k) - \hat{x}_\text{c}(k) = (I\!-\!LC)A \epsilon_\text{c}(k\!-\!1) + (I\!-\!LC)v(k\!-\!1) - Lw(k).
\label{eq:closedLoopCentralized_est}
\end{align}
We thus assume that $(I-LC)A$ is stable, which is always possible since $(A,C)$ is detectable.
It follows \cite[p.~212--213]{CaDe91} that there exist $m_\text{c}>0$ and $\rho_\text{c} \in [0,1)$ such that
\begin{equation}
\|((I-LC)A)^k \| \leq m_\text{c} \rho_\text{c}^k .
\label{eq:expStabCentral}
\end{equation}

\subsection{Problem statement}
\label{sec:objective}

The main objective of this article is 
an EBSE design that approximates the reference design of \sect \ref{sec:referenceDesign}
with guaranteed bounds:
%
\begin{problem}
\label{pro:EBSE}
Develop a distributed EBSE design for the scenario in \fig \ref{fig:networkArchitecture}, where each estimator agent ($i = \Nsens, \dots, \Nsens + \Nest$) locally computes an estimate $\hat{x}_i(k)$ of the state $x(k)$, and each sensor agent ($i = 1, \dots, \Nsens$) makes individual transmit decisions for its local measurements $y_i(k)$.  The design shall emulate the centralized estimator \eqref{eq:FCSE1}, \eqref{eq:FCSE2} bounding the difference $\|\hat{x}_\text{c}(k) - \hat{x}_i(k)\|$, but with reduced communication of sensor measurements. 
\end{problem}
%
Furthermore, we address distributed control based on the EBSE design:
\begin{problem}
\label{pro:EBC}
Design distributed control laws for computing control inputs $u_i(k)$ (\cf \eqref{eq:u_decomposed}) locally 
from the event-based estimates $\hat{x}_i(k)$ so as to 
achieve stable closed-loop dynamics (bounded $x$).
\end{problem}

For state estimation in general, both the measurement signal $y$ and the control input $u$ must be known (\cf \eqref{eq:FCSE1}, \eqref{eq:FCSE2}).  
For simplicity, we first focus on the reduction of sensor measurements and assume
\begin{assum}
\label{ass:knownInput}
The input $u$ is known by all agents. 
\end{assum} 
\noindent
This is the case, for example,  when estimating a process without control input (\ie $u=0$),
when $u$ is an a-priori known reference signal, or when $u$ is broadcast periodically over the shared bus.
In particular, if the components $u_i(k)$ 
are computed by different agents as in \Pro \ref{pro:EBC}, \Assump \ref{ass:knownInput} requires the agents to exchange their inputs over the bus at every step $k$.
Reducing measurement communication, but periodically exchanging inputs may be a viable solution when there are more measurements than control inputs (as is the case for the experiment presented in \sect \ref{sec:experimentsBC}).

Later, in \sect \ref{sec:control}, an extension of the results is presented, which does not require \Assump \ref{ass:knownInput} and periodic exchange of inputs by employing event-triggering protocols also for the inputs.
\section{Event-based state estimation with a single sensor-estimator link}
\label{sec:singleAgent}
In order to develop the main ideas of the EBSE approach, we first consider \Pro \ref{pro:EBSE} for the simpler, but relevant special case with $\Nsens = \Nest = 1$; that is, a single sensor transmits data over a network link to a remote estimator (also considered in \cite{Le11,MaSi10,SiLa12,WuJiJoSh13,TrCa15}, for instance).
%
For the purpose of this section, we make the simplifying assumption of a prefect communication link:
\begin{assum}
\label{ass:idealComm}
Transmission from sensor to estimator is instantaneous
and no data is  lost.
\end{assum}
\noindent
For a sufficiently fast network link,
this may be ensured by low-level protocols using acknowledgments and re-transmissions.  However, this assumption is made for the sake of simplicity in this section, and omitted again in the later sections.

We propose the event-based architecture depicted in \fig \ref{fig:EBSE_SingleAgent}. 
The key idea is to replicate the remote state estimator at the sensor; the sensor agents then knows what the estimator knows, and thus also when the estimator is in need of new data.  
The \emph{State Estimator} and \emph{Event Trigger},
which together form the EBSE algorithm, are explained next.

\begin{figure}[tb]
\centering
\subfigure[Single sensor/estimator agent]{%
    \label{fig:EBSE_SingleAgent}
    \includegraphics[width=0.65\columnwidth]{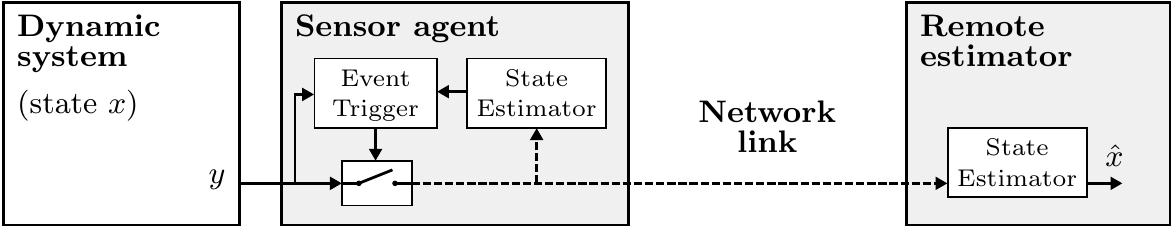}}
\\[5mm]
\subfigure[Multiple sensor/estimator agents]{%
    \label{fig:EBSE_MultiAgent}  
    \includegraphics[width=0.8\columnwidth]{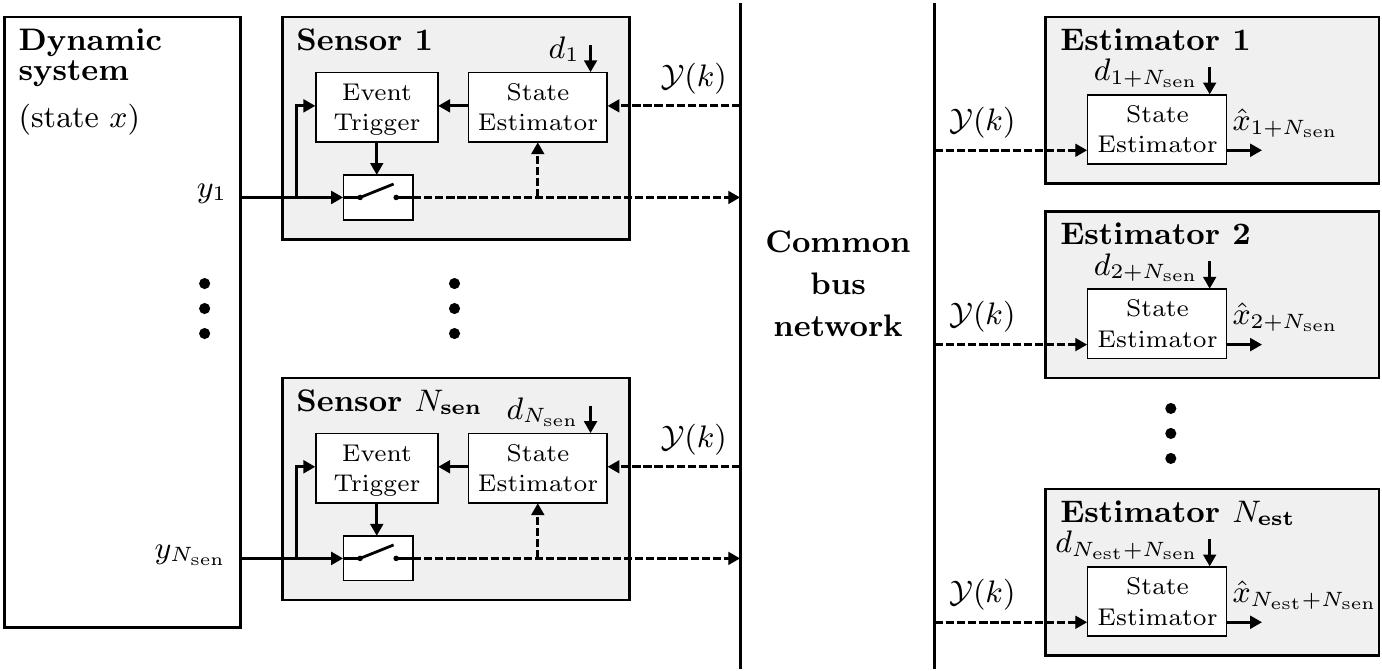}}
\caption{Proposed event-based state estimation architectures.  Dashed arrows indicate event-based communication, while solid ones indicate periodic communication.
\subcaption{(a)}{Single sensor/estimator case: The sensor agent implements a copy of the remote \emph{State Estimator} to trigger a data transmission (\emph{Event Trigger}) whenever an update is needed at the remote agent.}
\subcaption{(b)}{Multiple sensor/estimator case: Each agent implements a copy of the State Estimator for making transmit decisions (Sensors) or having full state information available (Estimators).  The common bus supports data exchange between all agents; $\mathcal{Y}(k)$ denotes the set of measurements communicated at time $k$.  Disturbances $d_i$ model differences in the agents' estimates, \eg from imperfect communication.}
}
\label{fig:EBSE_arch}
\end{figure}

\subsection{State estimator}
\label{sec:StateEstimator}
Both sensor and remote agent implement the state estimator (\cf \fig \ref{fig:EBSE_SingleAgent}).
The estimator recursively computes an estimate $\hat{x}_i(k) = \hat{x}_i(k|k)$ of the system state $x(k)$ from the available measurements:
\begin{align}
\hat{x}_i(k|k-1) &= A  \hat{x}_i(k-1|k-1) + B  u(k-1) \label{eq:EBSE1_single}  \\
\hat{x}_i(k|k) &= \hat{x}_i(k|k-1) + \gamma(k) L  \big(y(k) - C \, \hat{x}_i(k|k-1) \big)   \label{eq:EBSE2_single} 
\end{align}
with $i=1$ for the sensor, $i=2$ for the estimator, $L$ as in \eqref{eq:FCSE2},
and $\gamma(k) \in \{0,1\}$ denoting the sensor's decision of transmitting $y(k)$ ($\gamma(k)=1$), or not ($\gamma(k)=0$).  
%

By \Assump \ref{ass:idealComm}, both estimators have the same input data. 
 If, in addition, they are initialized identically,
 both estimates are identical, \ie $\hat{x}_1(k) = \hat{x}_2(k)$ for all $k$.  Hence, the sensor has knowledge about the estimator and can  exploit this for the triggering decision.

\subsection{Event trigger}
\label{sec:EventTriggerMeas}
The sensor transmits a measurement if, and only if, the remote estimator cannot predict the measurement accurately enough based on its state prediction.  
Specifically, $y(k)$ is transmitted when the remote prediction $\hat{y}(k) = C \hat{x}_2(k|k-1)$ deviates from $y(k)$ by more than a tolerable threshold $\deltae \geq 0$. 
Since $\hat{x}_1(k|k-1) = \hat{x}_2(k|k-1)$, the sensor can make this decision without requiring communication from the remote estimator:
\begin{equation}
\text{transmit $y(k)$} 
\; \Leftrightarrow \;
\text{$\gamma(k) = 1$} 
\; \Leftrightarrow \;
\| y(k) - C \hat{x}_1(k|k-1) \| \geq \deltae .
\label{eq:eventTrigger_MB}
\end{equation}
%
%
%
Tuning $\deltae$ allows the designer to trade off the sensor's frequency of events (and, hence, the communication rate) for estimation performance.
This choice of the trigger will be instrumental in bounding 
the difference between the event-based and the centralized estimator, as will be seen in the subsequent stability analysis.  
The trigger is also called \emph{innovation-based trigger} and was previously proposed in different contexts in \cite{YoTiSo02,TrDAn11,WuJiJoSh13}.
The innovation trigger \eqref{eq:eventTrigger_MB} can also be realized without the local state estimator on the sensor by periodically communicating estimates from the remote estimator to the sensor.  However, the proposed architecture
avoids this additional communication.  

\subsection{Stability analysis}
\label{sec:idealStabAnalysis}
The estimator update equations \eqref{eq:EBSE1_single}, \eqref{eq:EBSE2_single} and the triggering rule (\ref{eq:eventTrigger_MB}) together constitute the proposed event-based state estimator.   The estimator \eqref{eq:EBSE1_single}, \eqref{eq:EBSE2_single} is a switching observer, whose switching modes are governed by the event trigger \eqref{eq:eventTrigger_MB}.  
For arbitrary switching, stability of the switching observer is \emph{not} implied by stability of the centralized design (see \eg \cite{BoLu02}).  Hence, proving stability is an essential, non-trivial requirement for the event-based design.

\subsubsection{Difference to centralized estimator}
Addressing \Pro \ref{pro:EBSE}, we first prove a bounded difference to
the centralized reference estimator $\hat{x}_c(k)$.
Using \eqref{eq:FCSE1}, \eqref{eq:FCSE2}, \eqref{eq:EBSE1_single}, and \eqref{eq:EBSE2_single}, the difference $e_i(k) = \hat{x}_c(k) - \hat{x}_i(k)$ can be written as
\begin{align}
e_i(k) 
&=A e_i(k\!-\!1) 
+ L \big( y(k) - C \hat{x}_\text{c}(k|k\!-\!1) \big) - \gamma(k) L \big( y(k) - C \hat{x}_i(k|k-1) \big) \nonumber  \\
&=(I-LC)A e_i(k-1) 
+ (1-\gamma(k)) L \big( y(k) - C \hat{x}_i(k|k-1) \big) 
\label{eq:epsilon_ci}
\end{align}
where the last equation was obtained by adding and subtracting $L ( y(k) - C \hat{x}_i(k|k\!-\!1) )$.
The error $e_i(k)$ is governed by the stable centralized estimator dynamics $(I-LC)A$ 
with an extra input term, 
which is bounded by the choice of the event-trigger (\ref{eq:eventTrigger_MB}): for $\gamma(k)=0$, $L (y(k) - C \hat{x}_i(k|k-1))$ is bounded by (\ref{eq:eventTrigger_MB}), and for $\gamma(k)=1$, the extra term vanishes.
We thus have the following result:
\begin{theorem}
\label{thm:epsilon_ci_idealized}
Let Assumptions  \ref{ass:knownInput} and \ref{ass:idealComm} be satisfied, $(I-LC)A$ be stable, and $\hat{x}_1(0) = \hat{x}_2(0) = x_0$ for some $x_0 \in \R^n$.  Then, the difference $e_i(k)$ between the centralized estimator and the EBSE \eqref{eq:EBSE1_single}, \eqref{eq:EBSE2_single}, and \eqref{eq:eventTrigger_MB} is bounded by
\begin{equation}
\| e_i \|_\infty \leq m_\text{c} \| e_i(0) \| +  \frac{m_\text{c}}{1-\rho_\text{c}} \|L\| \delta^\text{\normalfont{est}} =: e_i^\text{max}.
\label{eq:thm_epsilon_ci_idealized}
\end{equation}
\end{theorem}
\begin{proof}
From the assumptions, it follows that $\hat{x}_1(k) = \hat{x}_2(k)$ and $\hat{x}_1(k|k-1) = \hat{x}_2(k|k-1)$.  From the previous argument, we have
\begin{equation}
\| (1-\gamma(k)) L \big( y(k) - C \hat{x}_i(k|k-1) \big) \|
\leq
\| L \| \delta^\text{\normalfont{est}} .
\end{equation}
The bound \eqref{eq:thm_epsilon_ci_idealized} then follows from \cite[p.~218, Thm. 75]{CaDe91} and exponential stability of $e_i(k) = (I-LC)A e_i(k-1)$ (\cf \eqref{eq:expStabCentral}).
\end{proof}

The first term in \eqref{eq:thm_epsilon_ci_idealized}, $m_\text{c} \| e_i(0) \|$, is due to possibly different initial conditions between the EBSE and the centralized estimator, 
and $m_\text{c} \|L\| \deltae / (1-\rho_\text{c})$ represents the asymptotic bound.  Choosing $\deltae$ small enough, $e_i(k)$ can hence be made arbitrarily small as $k \to \infty$,
and, for $\deltae = 0$,
the performance of the centralized estimator is recovered.

The bound (\ref{eq:thm_epsilon_ci_idealized}) holds irrespective of the nature of the disturbances $v$ and $w$ in \eqref{eq:system_x}, \eqref{eq:system_y} (no assumption on $v$, $w$ is made in \The \ref{thm:epsilon_ci_idealized}).  In particular, it also holds for the case of unbounded disturbances such as Gaussian noise.  

\subsubsection{Estimation error}
\label{sec:estErrSingle}
The actual estimation error $\epsilon_i$ of agent $i$ is
\begin{align}
\epsilon_i(k) &:= x(k) - \hat{x}_i(k)
= \epsilon_\text{c}(k) + e_i(k) . 
\label{eq:estError_ei}
\end{align}
\The \ref{thm:epsilon_ci_idealized} can be used to deduce properties of the estimation error $\epsilon_i$ from properties of the centralized estimator.
We exemplify this for the case of bounded, as well as stochastic disturbances $v$ and $w$.
\begin{corollary}
\label{cor:estErrDet}
Let $\|v\|_\infty \leq v^\text{max}$, $\|w\|_\infty \leq w^\text{max}$, $\| e_i \|_\infty \leq e_i^\text{max}$ be bounded, and $(I-LC)A$ be stable.  Then, 
the event-based estimation error \eqref{eq:estError_ei} is bounded by
\begin{equation}
\| \epsilon_i \|_\infty \leq \epsilon_\text{c}^\text{max} + e_i^\text{max}
\label{eq:bound_epsiloni_deterministic}
\end{equation}
with $\epsilon_\text{c}^\text{max}:= m_\text{c} \| \epsilon_\text{c}(0) \| +  \frac{m_\text{c}}{1-\rho_\text{c}} (\|I-LC\| v^\text{\normalfont{max}} + \|L\| w^\text{\normalfont{max}})$.
\end{corollary}
\begin{proof}
The bound $\epsilon_\text{c}^\text{max}$ on the centralized estimation error $\epsilon_\text{c}(k)$ follows directly from \eqref{eq:closedLoopCentralized_est}, exponential stability \eqref{eq:expStabCentral}, and \cite[p.~218, Thm. 75]{CaDe91}.  The result \eqref{eq:bound_epsiloni_deterministic} is then immediate from \eqref{eq:estError_ei}.
\end{proof}
\begin{corollary}
\label{cor:estErrStoch}
Let $v$, $w$, $x(0)$ be random variables with $\E[v(k)] = 0$, $\E[w(k)] = 0$, $\E[x(0)] = x_0$, and the centralized estimator be initialized with $\hat{x}_\text{c}(0) = x_0$.  Let $\| e_i \|_\infty \leq e_i^\text{max}$ be bounded, and $(I-LC)A$ be stable.  
Then, the expected event-based estimation error \eqref{eq:estError_ei} is bounded by
\begin{equation}
\| \E[\epsilon_i(k)] \| \leq e_i^\text{max} \quad \forall k.
\label{eq:bound_epsiloni_stochastic}
\end{equation}
\end{corollary}
\begin{proof}
From \eqref{eq:closedLoopCentralized_est}, it follows $\E[ \epsilon_\text{c}(k) ] = (I-LC)A \E[ \epsilon_\text{c}(k-1) ]$, and thus $\E[ \epsilon_\text{c}(k) ] = 0$ by recursion from $\E[ \epsilon_\text{c}(0) ] = \E[x(0)] - x_0 = 0$.  Therefore,
\begin{align}
\| \E[\epsilon_i(k)] \| 
&= \| \E[e_i(k)] \| 
\leq \E[ \|e_i(k) \| ]
\leq e_i^\text{max}
\end{align}
where the first inequality follows from Jensen's inequality, and the last from $\| e_i(k) \| \leq e_i^\text{max}$.
\end{proof}

\section{Event-based state estimation with multiple agents}
\label{sec:multiAgent}
We extend the ideas of the previous section to the general multi-agent case in \Pro \ref{pro:EBSE}.
%
While the assumption of perfect communication (\Assump \ref{ass:idealComm}) may possibly be realizable for few agents, it becomes unrealistic as the number of agents increases.  Thus, we generalize the stability analysis to the case where agents' estimates can differ.

\subsection{Architecture}
We propose the distributed event-based architecture depicted in \fig \ref{fig:EBSE_MultiAgent} for the multi-agent problem.
Adopting the key idea of the single agent case (\cf \fig \ref{fig:EBSE_SingleAgent}), each agent implements a copy of the state estimator for making transmit decisions.  
The common bus network ensures that, if a measurement is transmitted, it is broadcast to all other units.  For the  estimators to be consistent, the sensor agents also listen to the measurement data $\mathcal{Y}(k)$ broadcast by other units.  
%

The proposed EBSE scheme is \emph{distributed} in the sense that data from distributed sensors is required for stable state estimation, and that transmit decisions are made locally by each agent.

\subsection{Event trigger}
\label{sec:EventTriggerMeas_Multi}
In analogy to the single agent case \eqref{eq:eventTrigger_MB}, each agent $i$, $i \in \N_{\Nsens}$, uses the following event triggering rule:
\begin{equation}
\text{transmit $y_i(k)$} 
\; \Leftrightarrow \;
\| y_i(k) - C_i \hat{x}_i(k|k-1) \| \geq \deltae_i .
\label{eq:eventTrigger_MB_multi}
\end{equation}
The prediction $\hat{y}_i(k) = C_i \hat{x}_i(k|k-1)$ computed by agent $i$ is representative also for all other agents' predictions  of the same measurement, $\hat{y}_{i}^j(k) = C_i \hat{x}_j(k|k-1)$, as long as $\hat{x}_i(k|k-1) \approx \hat{x}_j(k|k-1)$, which is to be established in the stability analysis below.  
Being able to approximately represent the other agents' knowledge
is the basis for making effective transmit decisions in the proposed approach.

For later reference, we introduce $\deltae := (\deltae_1, \dots, \deltae_{\Nsens})$ and the index sets of transmitting and not-transmitting agents:
\begin{align}
I(k) &:= \{ i \in \N_{\Nsens} \, | \, \| y_i(k) - C_i \hat{x}_i(k|k\!-\!1) \| \geq \deltae_i \}
\label{eq:I} \\
\Ic(k) &:= \{ i \in \N_{\Nsens} \, | \, \| y_i(k) - C_i \hat{x}_i(k|k\!-\!1) \| < \deltae_i \}  = \N_{\Nsens} \setminus I(k).
\label{eq:Ibar}
\end{align}

\subsection{State estimator}
Extending the event-based estimator \eqref{eq:EBSE1_single}, \eqref{eq:EBSE2_single} to the multi sensor case, we propose the following estimator update for all agents ($i \in \mathbb{N}_\Nall$):
\begin{align}
\hat{x}_i(k|k-1) &= A  \hat{x}_i(k-1|k-1) + B  u(k-1) \label{eq:EBSE1_multi}  \\
\hat{x}_i(k|k) &= \hat{x}_i(k|k-1) + \!\! \sum_{\ell \in I(k)} \!\! L_\ell \big( y_\ell(k) - C_\ell \hat{x}_i(k|k-1) \big)   \label{eq:EBSE2ideal_multi} 
\end{align}
where $L_\ell \in \R^{n \times p_\ell}$ is the submatrix of the centralized gain  $L = [L_1, L_2, \dots, L_{\Nsens}]$ corresponding to $y_\ell$.
Rewriting \eqref{eq:FCSE2} as 
\begin{align}
\hat{x}_\text{c}(k|k) &= \hat{x}_\text{c}(k|k-1) + \!\! \sum_{\ell\in \N_{\Nsens}} \!  L_\ell \big( y_\ell(k) - C_\ell \hat{x}_\text{c}(k|k-1) \big) 
\label{eq:FCSE2_rewritten} 
\end{align}
we see that \eqref{eq:EBSE2ideal_multi} is the same as the centralized update, but only updating with a subset $I(k) \subset \N_{\Nsens}$ of all measurements.
If, at time $k$, no measurement is transmitted (\ie $I(k) = \emptyset$), then the summation in (\ref{eq:EBSE2ideal_multi}) 
vanishes; 
that is, $\hat{x}_i(k|k) = \hat{x}_i(k|k-1)$. 

To account for differences in any two agents' estimates, \eg from unequal initialization, different computation accuracy, or imperfect communication, 
we introduce a generic disturbance signal $d_i$ acting on each estimator (\cf \fig \ref{fig:EBSE_MultiAgent}).  For the stability analysis, we thus replace \eqref{eq:EBSE2ideal_multi} with
\begin{align}
\hat{x}_i(k|k) &= \hat{x}_i(k|k\!-\!1) \! + \! \sum\limits_{\ell \in I(k)} \!\! L_\ell \big( y_\ell(k) - C_\ell \hat{x}_i(k|k\!-\!1) \big) +d_i(k).  \label{eq:EBSE2_multi}
\end{align}
%
The disturbances are assumed to be bounded:
\begin{assum}
\label{ass:bounded_di}
For all $i \in \N_{\Nall}$, $\|d_i \|_\infty \leq d_i^\text{\normalfont{max}}$.
\end{assum}
\noindent
This assumption is realistic, when $d_i$ represent imperfect initialization or different computation accuracy, for example.
Even though the assumption may not hold for modeling packet drops in general,
the developed method was found to be effective also for this case in the example of \sect \ref{sec:simulationExample}. 
%

\subsection{Stability analysis}
We discuss stability of the distributed EBSE system given by the process \eqref{eq:system_x}, \eqref{eq:system_yi}, the (disturbed) estimators \eqref{eq:EBSE1_multi}, \eqref{eq:EBSE2_multi}, and the triggering rule \eqref{eq:eventTrigger_MB_multi}.
We first consider the difference between the centralized and event-based estimate, $e_i(k) = \hat{x}_c(k) - \hat{x}_i(k)$.  
By straightforward manipulation using \eqref{eq:FCSE1}, \eqref{eq:EBSE1_multi}, \eqref{eq:FCSE2_rewritten},  and \eqref{eq:EBSE2_multi}, we obtain
\begin{align}
e_i(k) 
&=A e_i(k\!-\!1) 
+ \sum\limits_{\ell \in \N_{\Nsens}} L_\ell \big( y_\ell(k) - C_\ell \hat{x}_\text{c}(k|k\!-\!1) \big) 
-\!\!\underbrace{\sum\limits_{\ell \in I(k)}  L_\ell \big( y_\ell(k) - C_\ell \hat{x}_i(k|k-1) \big)}_{
\sum_{\ell \in \N_{\Nsens}} \! L_\ell (\, \ldots \,) \,\,
- \,\, \sum_{\ell \in \Ic(k)} \! L_\ell (\, \ldots \,)
} -d_i(k) \nonumber  \\
&=(I-LC)A e_i(k-1) +\!\! \sum\limits_{\ell \in \Ic(k)}  L_\ell \big( y_\ell(k) - C_\ell \hat{x}_i(k|k-1) \big)  -d_i(k)
\label{eq:epsilon_ci_multi} \\
%
&=(I-LC)A e_i(k-1) \nonumber \\
&\phantom{=}+\!\! \sum\limits_{\ell \in \Ic(k)}  L_\ell \big( y_\ell(k) - C_\ell \hat{x}_\ell(k|k-1) \big) -d_i(k) -\!\! \sum\limits_{j \in \Ic(k)} L_j C_j A e_{ij}(k-1) 
\label{eq:epsilon_ci_di}
\end{align}
where $e_{ij}(k) := \hat{x}_i(k) - \hat{x}_j(k)$ is the inter-agent error, 
and we used $\hat{x}_i(k|k-1) - \hat{x}_\ell(k|k-1) = A e_{i \ell}(k-1)$.
The error dynamics \eqref{eq:epsilon_ci_di} are governed by stable dynamics $e_i(k) =(I-LC)A e_i(k-1)$ with three input terms.
The term $\sum_{\ell \in \Ic(k)}  L_\ell ( y_\ell(k) - C_\ell \hat{x}_\ell(k|k-1) )$ is analogous to the last term in \eqref{eq:epsilon_ci} and bounded by the event triggering \eqref{eq:eventTrigger_MB_multi} (\cf \eqref{eq:Ibar}).  The last two terms are due to the disturbance $d_i$ and resulting inter-agent differences $e_{ij}$. 
To bound $e_i$, $e_{ij}$ must also be bounded, which is established next. 

\subsubsection{Inter-agent error}
\label{sec:interAgentError}
The inter-agent error can be written as
\begin{align}
e_{ij}(k) 
&= \hat{x}_i(k) - \hat{x}_j(k) 
=A e_{ij}(k-1) 
\nonumber \\
&\phantom{=} + \sum\nolimits_{\ell \in I(k)} L_\ell \big( y_\ell(k) - C_\ell \hat{x}_i(k|k-1) \big) +d_i(k) \nonumber \\
&\phantom{=} - \sum\nolimits_{\ell \in I(k)} L_\ell \big( y_\ell(k) - C_\ell \hat{x}_j(k|k-1) \big) -d_j(k) \nonumber \\
&= 
\tilde{A}_{I(k)} \, e_{ij}(k-1)
+d_i(k) -d_j(k) \label{eq:epsij_dyn} 
\end{align}
where $\tilde{A}_{J}$ is defined for some subset $J \subseteq \N_{\Nsens}$ by
\begin{equation}
\tilde{A}_{J} := (I - \sum_{\ell \in J} L_\ell  C_\ell)A .
\end{equation}
Hence, the inter-agent error $e_{ij}(k)$ is governed by the time-varying dynamics $e_{ij}(k) = \tilde{A}_{I(k)} e_{ij}(k-1)$. Unfortunately, one cannot, in general, infer stability of the inter-agent error (and thus the event-based estimation error \eqref{eq:epsilon_ci_di}) from stability of the centralized design. 
A counterexample is presented in \cite{Tr14}.

A sufficiency result for stability of the inter-agent error can be obtained by considering the dynamics \eqref{eq:epsij_dyn} under arbitrary switching; that is, with $\tilde{A}_{J}$ for all subsets $J \subseteq \N_{\Nsens}$.  The following result is adapted from \cite[Lemma~3.1]{MuTr15}.
\begin{lemma}
Let \Assump \ref{ass:bounded_di} hold, and let the matrix inequality 
\begin{equation}
\tilde{A}_{J}^\transpose P \tilde{A}_{J} - P < 0
\label{eq:LMI_cond}
\end{equation}
be satisfied for some positive definite $P \in \R^{n \times n}$ and for all subsets $J \subseteq \N_{\Nsens}$.
Then, for given initial errors $e_{ij}(0)$ ($i,j \in \N_N$), there exists $e^\text{max} \in \R$, $e^\text{max} \geq 0$, such that 
\begin{equation}
\| e_{ij} \|_\infty \leq e^\text{max}, \quad \text{for all $i,j \in \N_N$ and the Euclidean norm $\| . \|$.} 
\label{eq:bounded_eij}
\end{equation}
\label{lem:bounded_eij}
\end{lemma}
\begin{proof}
Under \eqref{eq:LMI_cond}, the error dynamics \eqref{eq:epsij_dyn} are input-to-state stable (ISS) following the proof of \cite[Lemma~3.1]{MuTr15} ($A_\text{cl}(J)$ replaced with $\tilde{A}_{J}$).
With \Assump \ref{ass:bounded_di}, ISS guarantees boundedness of the inter-agent error $e_{ij}$ and thus the existence of $e^\text{max}_{ij} \geq 0$ (possibly dependent on the initial error $e_{ij}(0)$) such that 
\begin{equation}
\| e_{ij} \|_\infty \leq e^\text{max}_{ij} .
\end{equation}
Finally, \eqref{eq:bounded_eij} is obtained by taking the maximum over all $e^\text{max}_{ij}$.
\end{proof}

%


The stability test is conservative because the event trigger \eqref{eq:eventTrigger_MB_multi} will generally not permit arbitrary switching.  Since $J \subseteq \N_{\Nsens}$ also includes the empty set (\ie $\tilde{A}_{\emptyset} = A$), the test can only be used for open-loop stable dynamics \eqref{eq:system_x}.
In \sect \ref{sec:syncAvg}, we present an alternative approach to obtained bounded $e_{ij}$ for arbitrary systems.

\subsubsection{Difference to centralized estimator}
With the preceding lemma, we can now establish boundedness of the estimation error \eqref{eq:epsilon_ci_di}.
\begin{theorem}
\label{thm:epsilon_ci_notIdealized}
Let Assumptions \ref{ass:knownInput} and \ref{ass:bounded_di} and the conditions of \Lem \ref{lem:bounded_eij} be satisfied,
and let $(I-LC)A$ be stable.
Then, the difference $e_i(k)$ between the centralized estimator and the EBSE \eqref{eq:eventTrigger_MB_multi}, \eqref{eq:EBSE1_multi}, \eqref{eq:EBSE2_multi}  is bounded by
\begin{equation}
\| e_i \|_\infty \leq m_\text{c} \| e_i(0) \| 
+  \frac{m_\text{c}}{1-\rho_\text{c}} 
\big( 
\|L\| \|\delta^\text{\normalfont{est}} \|
+ d_i^\text{\normalfont{max}}
+  \bar{m} \Nsens e^\text{max}
\big)
=: e_i^{\text{\normalfont{max}}} 
\label{eq:thm_epsilon_ci_notIdealized}
\end{equation}
with $m_\text{c}$, $\rho_\text{c}$ as in \eqref{eq:expStabCentral}, and $\bar{m} := \max_{j \in \N_{\Nsens}} \| L_jC_jA \|$.
\end{theorem}
\begin{proof}
We can establish the following bounds (for all $k$)
\begin{align}
&\Big\| \sum\limits_{\ell \in \Ic(k)}  L_\ell \big( y_\ell(k) - C_\ell \hat{x}_\ell(k|k-1) \big) \Big\| 
\stackrel{\text{\eqref{eq:Ibar}}}{\leq}  \|L\| \|\delta^\text{\normalfont{est}} \| \\
&\| d_i(k) \| 
\stackrel{\text{Ass.~\ref{ass:bounded_di}}}{\leq} d_i^\text{max} \\
&\Big\| \sum\limits_{j \in \Ic(k)} L_j C_j A e_{ij}(k-1) \Big\|
\leq \sum\limits_{j \in \Ic(k)} \|L_j C_j A\| \|e_{ij}(k-1)\| 
\stackrel{\text{Lemma~\ref{lem:bounded_eij}}}{\leq} \bar{m} \Nsens e^\text{max} . 
\end{align}
The result \eqref{eq:thm_epsilon_ci_notIdealized} then follows from \eqref{eq:epsilon_ci_di}, stability of $(I-LC)A$, and \cite[p.~218, Thm. 75]{CaDe91}.
\end{proof}

\subsubsection{Synchronous estimator resets}
\label{sec:syncAvg}
We present a straightforward extension of the event-based communication scheme, which guarantees stability even if the inter-agent error dynamics \eqref{eq:epsij_dyn} cannot be shown to be stable (e.g., if \Lem \ref{lem:bounded_eij} does not apply).

Since the inter-agent error $e_{ij}(k)$ is the difference between the state estimates by agent $i$ and $j$, 
we can make it zero 
by resetting the two agents' state estimates to the same value, for example, their average. 
Therefore, a straightforward way to guarantee 
bounded inter-agent errors is to periodically reset all agents' estimates to their joint average.
Clearly, this strategy increases the communication load on the network.  If, however, the disturbances $d_i$ are small or only occur rarely, the required resetting period can typically be large relative to the underlying sampling time $\Ts$.  

We assume that the resetting happens after all agents have made their estimator updates \eqref{eq:EBSE2_multi}.
Let $\hat{x}_i(k-)$ and $\hat{x}_i(k+)$ denote agent $i$'s estimate at time $k$ before and after resetting, and let $K \in \N$ be the fixed resetting period.  Each agent $i$ implements the following synchronous averaging: 
\begin{align}
&\text{If $k$ a multiple of $K$:} 
\!\!\!&& \text{transmit $\hat{x}_i(k-)$;} \label{eq:syncAvg} \\
& && \text{receive $\hat{x}_j(k-), j \in \N_N \!\setminus \!\{i \}$;} \nonumber \\
& && \text{set $\hat{x}_i(k+) = \sum_{j=1}^{\Nall} \hat{x}_j(k-)$.} \nonumber
\end{align}
%
We assume that the network capacity is such that the mutual exchange of the estimates can happen in one time step,
and no data is lost in the transfer.
%
In other scenarios, one could take several time steps to exchange all estimates, at the expense of a delayed reset. 
%
The synchronous averaging period $K$ can be chosen from simulations assuming a model for the inter-agent disturbances $d_i$ (\eg packet drops).  


We have the following stability result for EBSE with synchronous averaging \eqref{eq:syncAvg}.
\begin{theorem}
\label{thm:epsilon_ci_syncAvg}
Let Assumptions \ref{ass:knownInput} and \ref{ass:bounded_di} be satisfied and $(I-LC)A$ be stable. Then, the difference $e_i(k)$ between the centralized estimator and the EBSE with synchronous averaging given by \eqref{eq:eventTrigger_MB_multi}, \eqref{eq:EBSE1_multi}, \eqref{eq:EBSE2_multi},  and \eqref{eq:syncAvg}  is bounded.
\end{theorem}
\begin{proof} 
Since the agent error \eqref{eq:epsilon_ci_di} is affected by the resetting \eqref{eq:syncAvg}, we first rewrite $e_i(k)$ in terms of the average estimate $\bar{x}(k) := \avg( \hat{x}_i(k) ) := \frac{1}{N} \sum\nolimits_{i=1}^N \hat{x}_i(k)$.  Defining $\bar{e}(k) := \hat{x}_\text{c}(k) - \bar{x}(k)$ and $\bar{e}_i(k) := \bar{x}(k) - \hat{x}_i(k)$, we have $e_i(k) = \bar{e}(k) + \bar{e}_i(k)$ and will establish the claim by showing boundedness of $\bar{e}_i(k)$ and $\bar{e}(k)$ separately.

For the average estimate $\bar{x}(k)$,
we obtain from \eqref{eq:EBSE1_multi}, \eqref{eq:EBSE2_multi},
\begin{align*}
&\bar{x}(k|k-1) = A \bar{x}(k-1) + B  u(k-1) \\
&\bar{x}(k) = \bar{x}(k|k\!-\!1) + \!\!\! \sum_{\ell \in I(k)} \!\! L_\ell \big( y_\ell(k) - C_\ell \bar{x}(k|k\!-\!1) \big) + \bar{d}(k)
\end{align*}
where $\bar{x}(k|k-1) := \avg(\hat{x}_i(k|k-1))$
and $\bar{d}(k) := \avg(d_i(k))$.
The dynamics of the error $\bar{e}_i(k)$ are described by 
\begin{align}
\bar{e}_i(k) 
&= 
\tilde{A}_{I(k)} \bar{e}_i(k-1) 
+\bar{d}(k)-d_i(k) 
 \label{eq:epsj}  \\
\bar{e}_i(k+) &= 0, \qquad \text{for $k = \kappa K$ with some $\kappa \in \N$} \label{eq:epsj_reset} 
\end{align}
where \eqref{eq:epsj} is obtained by direct calculation analogous to \eqref{eq:epsij_dyn}, and \eqref{eq:epsj_reset} follows from \eqref{eq:syncAvg}.  
Since $\bar{d}(k)$, $d_i(k)$, 
and $\tilde{A}_{I(k)}$ are all bounded, boundedness of  $\bar{e}_i$ for all $i \in \N_N$ follows.

Since $\bar{e}(k) = \avg_i(e_i(k))$, we obtain from \eqref{eq:epsilon_ci_di}
\begin{align}
\bar{e}(k)
&= (I-LC)A \bar{e}(k-1)   \nonumber \\
&\phantom{=}+ \sum\limits_{\ell \in \Ic(k)}  L_\ell \big( y_\ell(k) - C_\ell \hat{x}_\ell(k|k-1) \big) -\bar{d}(k) \nonumber \\
&\phantom{=}- \sum\limits_{j \in \Ic(k)}  L_j C_j A \bar{e}_{j}(k-1) 
\label{eq:eavg}
\end{align}
where 
we used $\avg_i(e_{ij}(k)) = \avg_i (\hat{x}_i(k) - \hat{x}_j(k)) = \bar{x}(k) - \hat{x}_j(k) = \bar{e}_j(k)$.
Note that \eqref{eq:eavg} fully describes the evolution of $\bar{e}(k)$.  In particular, the resetting \eqref{eq:syncAvg} does not affect $\bar{e}(k)$ because, at $k=\kappa K$, it holds 
\begin{align}
\bar{e}(k+) &= \hat{x}_\text{c}(k) - \frac{1}{N} \sum_{j=1}^N \hat{x}_j(k+) \nonumber \\
&=\hat{x}_\text{c}(k) - \frac{1}{N} \sum_{j=1}^N \Big( \frac{1}{N} \sum_{\ell=1}^N \hat{x}_\ell(k-) \Big) \nonumber \\
&=\hat{x}_\text{c}(k) - \frac{1}{N} \sum_{\ell=1}^N \hat{x}_\ell(k-) = \bar{e}(k-) .
\end{align}
%
All input terms in \eqref{eq:eavg} are bounded: $\bar{d}$ by \Assump \ref{ass:bounded_di}, $\sum_{\ell \in \Ic(k)} L_\ell (y_\ell(k) - C_\ell \hat{x}_\ell(k|k-1))$ by \eqref{eq:Ibar},
and $\bar{e}_j$ 
by the previous argument.  
The claim then follows from stability of $(I-LC)A$.
\end{proof}

\subsubsection{Estimation error}
\label{sec:estErrMulti}
By means of \eqref{eq:estError_ei} with \The \ref{thm:epsilon_ci_notIdealized} or \The \ref{thm:epsilon_ci_syncAvg}, properties about the agent's estimation error $\epsilon_i(k) = x(k) - \hat{x}_i(k)$ can be derived given properties of the disturbances $v$, $w$, and the centralized estimator.  
For example, Corollaries \ref{cor:estErrDet} and \ref{cor:estErrStoch} apply analogously also for the multi-agent case.

\section{Distributed control}
\label{sec:control}
In this section, we address \Pro \ref{pro:EBC}; that is,
the scenario where the local estimates $\hat{x}_i$ on the $\Nest$ estimators are used for feedback control.  

Recall the decomposition \eqref{eq:u_decomposed} of the control input, where $u_i(k)$ is the input computed on estimator agent $i + \Nsens$.
Assume a centralized state-feedback design is given
\begin{equation}
u(k) = F \, x(k)
\label{eq:stateFeedback}
\end{equation}
with controller gain $F \in \R^{q \times n}$ such that $A+BF$ is stable.  
We propose the distributed state-feedback control law
\begin{equation}
u_i(k) = F_i \, \hat{x}_{i+\Nsens}(k), \quad i \in \N_{\Nest}
\label{eq:stateFeedbackDistributed}
\end{equation}
where $F_i \in \R^{q_i \times n}$ is the part of the gain matrix $F$ in \eqref{eq:stateFeedback} corresponding to the local input $u_i$.  Same as for the emulation-based estimator design in previous sections, the feedback gains do not need to be specifically designed, but can simply be taken from the centralized design \eqref{eq:stateFeedback}.

\subsection{Closed-loop stability analysis}
Using \eqref{eq:estError_ei} and \eqref{eq:stateFeedbackDistributed}, the state equation \eqref{eq:system_x} can be rewritten as
\begin{align}
x(k) &= (A+BF) x(k-1) - \! \sum_{i \in \N_{\Nest}} \! B_i F_i \, \epsilon_{i+\Nsens}(k-1) + v(k-1)
\label{eq:stateEquationWFeedback}
\end{align}
where $\epsilon_{i+\Nsens}(k-1)$ are the estimation errors of the estimator agents (\cf \sect \ref{sec:estErrMulti}). Closed-loop stability can then be deduced leveraging the results of \sect \ref{sec:multiAgent}.
\begin{theorem}
\label{thm:boundednessControl}
Let the assumptions of either \The \ref{thm:epsilon_ci_notIdealized} or \The \ref{thm:epsilon_ci_syncAvg} be satisfied, $A+BF$ be stable, and $v$ and $w$ bounded.
Then, the state 
of the closed-loop
control system given by \eqref{eq:system_x}, \eqref{eq:system_y}, \eqref{eq:eventTrigger_MB_multi}, \eqref{eq:EBSE1_multi},  \eqref{eq:EBSE2_multi}, \eqref{eq:stateFeedbackDistributed}, and (possibly) \eqref{eq:syncAvg}
is bounded. 
\end{theorem}
\begin{proof}
Since $(I-LC)A$ is stable and $v$, $w$ bounded, it follows from \eqref{eq:closedLoopCentralized_est} that the estimation error $\epsilon_\text{c}(k)$ of the centralized observer is also bounded.  Thus, \eqref{eq:estError_ei} and \The \ref{thm:epsilon_ci_notIdealized} or \ref{thm:epsilon_ci_syncAvg} imply that all  $\epsilon_i$, $i \in \mathbb{N}_\Nall$, are bounded. 
Hence, it follows from \eqref{eq:stateEquationWFeedback}, stability of $A+BF$, and bounded $v$ that $x$ is also bounded.
\end{proof}

Satisfying \Assump \ref{ass:knownInput} for the above result requires the periodic communication of all inputs over the bus.  
While this increases the network load, it can be a viable option if the number of inputs is comparably small.  Next, we briefly present an alternative scheme, where the communication of inputs is reduced also by means of event-based protocols.

\subsection{Event-based communication of inputs}

Each estimator agent computes 
$u_i(k)$ according to \eqref{eq:stateFeedbackDistributed}
and broadcasts an update to the other agents whenever there has been a significant change:
\begin{equation}
\text{transmit $u_i(k)$} \;\; 
\Leftrightarrow \;\;
\|u_i(k) -  u_{i, \text{last}}(k) \| \geq \deltac_i
\label{eq:eventTriggerCtrl}
\end{equation}
where $\deltac_i \geq 0$ is a tuning parameter, and $u_{i, \text{last}}(k)$ is the last input that was broadcast by agent $i$. 
%

Each agent $i$ maintains an estimate $\hat{u}^i(k) \in \R^q$ of the complete input vector $u(k)$;
agent $i$'s estimate of agent $j$'s input is 
\begin{equation}
\hat{u}^i_j(k)
=
\begin{cases}
u_j(k) & \text{if \eqref{eq:eventTriggerCtrl} triggered} \\
\hat{u}^i_j(k-1) & \text{otherwise}.
\end{cases}
\label{eq:inputEstSingle}
\end{equation}
The agent then uses $\hat{u}^i(k-1) = ( \hat{u}^i_1(k-1), \hat{u}^i_2(k-1), \dots, \hat{u}_N(k-1) )$ 
instead of the true input $u(k-1)$ for the estimator update \eqref{eq:EBSE1_multi}. 
Since the error   $\tilde{u}^i(k) := u(k) - \hat{u}^i(k)$ from making this approximate update is bounded by the event trigger \eqref{eq:eventTriggerCtrl}, the stability results presented in \sect \ref{sec:multiAgent} can be extended to this case. 
The details are omitted, but can be found in \cite{Tr15arxiv}.


\section{Experiments}
\label{sec:experiments}
To illustrate the proposed approach for event-based state estimation, we present  numerical simulations of a benchmark problem \cite{GrHiJuEtAl14} and summarize experimental results from \cite{Tr12}.
Both are examples of the multi-agent case where local estimates are used for control.

\subsection{Numerical simulation of a thermo-fluid benchmark process}
\label{sec:simulationExample}
We consider distributed event-based control of a thermo-fluid process, which has been proposed as a benchmark problem in \cite{GrHiJuEtAl14,SiStGrLu15}.
%
Matlab files to run the simulation example are provided as supplementary material (\url{http://tiny.cc/DEBSEsuppl}).

The process has two tanks containing fluids, whose level and temperature are to be regulated by controlling the tanks' inflows, as well as heating and cooling units.  Both tanks are subject to step-like disturbances, and their dynamics are coupled through cross-flows between the tanks.  Each tank is associated with a control agent responsible for computing commands to the respective actuators.  Each agent can sense the temperature and level of its tank. For details on the process, refer to \cite{GrHiJuEtAl14,SiStGrLu15}.

\subsubsection{System description}
\label{sec:SystemDescriptionThermofluid}
The discrete-time linear model \eqref{eq:system_x}, \eqref{eq:system_y} 
is obtained by zero-order hold discretization  with $T_\text{s} = 0.2 \, \text{s}$ of the continuous-time model given in \cite[Sec.~5.8]{GrHiJuEtAl14}.
The process dynamics are stable.
The states and inputs of the system are summarized in Table~\ref{tab:StatesAndInputsThermofluid}.
%
Noisy state measurements 
\begin{equation*}
y(k) = x(k) + w(k)
\end{equation*}
are available, 
where $w(k)$ is uniformly distributed.
The numerical parameters of the model are available in the supplementary files.

\begin{table}
\renewcommand{\arraystretch}{1.1}
\caption{States and inputs of the thermo-fluid process.
}
\label{tab:StatesAndInputsThermofluid}
\centering
\begin{tabular}{lll}
\hline 
{\bf States} && {\bf Unit}   \\ 
\hline 
$x_1(k)$ & level tank 1  & m  \\
$x_2(k)$ & temperature tank 1 & K   \\
$x_3(k)$ & level tank 2  & m  \\
$x_4(k)$ & temperature tank 2 & K   \\ 
\hline 
{\bf Inputs} && {\bf Unit}   \\ 
\hline 
$u_{11}(k)$ & inflow tank 1  & 1 (normalized) \\
$u_{12}(k)$ & cooling tank 1  & 1 (normalized) \\
$u_{21}(k)$ & inflow tank 2  & 1 (normalized) \\
$u_{22}(k)$ & heating tank 2  & 1 (normalized) \\ 
\hline 
\end{tabular} 
\end{table}


Similar to the distributed architecture in \cite{SiStGrLu15}, we consider two agents (one for each tank) exchanging data with each other over a network link, see \fig \ref{fig:EBC_MultiAgent} (with $N=2$) and \tab \ref{tab:AgentsThermofluid} for inputs/output definitions.  
Each agents combines the functions of sensing, estimation, and control.  To save computational resources, an agent runs a single estimator and uses it for both event triggering \eqref{eq:eventTrigger_MB_multi} and feedback control \eqref{eq:stateFeedbackDistributed} (see \cite{TrDAn11} for an alternative architecture with two estimators).
%
\begin{table}[tb]
\caption{Agents in the thermo-fluid example.  Agent 1 measures $y_1 = (y_{11}, y_{12})$ (level and temperature) of its tank and is responsible for controlling $u_1 = (u_{11}, u_{12})$ (inflow and cooling); and Agent 2 accordingly.
%
}  
\label{tab:AgentsThermofluid}
\centering
%
%
%
\begin{tabular}{|l|ll|}
\hline 
{\bf Agent \#} & 1 & 2 \\ \hline 
{\bf Actuator}  & $u_1 = (u_{11}, u_{12})$ & $u_2 = (u_{21}, u_{22})$ \\ \hline
{\bf Sensors}  & $y_1 = (y_{11}, y_{12})$ & $y_2 = (y_{21}, y_{22})$  \\ \hline 
\end{tabular} 
\end{table}

\begin{figure}[tb]
\centering
\includegraphics[width=0.55\columnwidth]{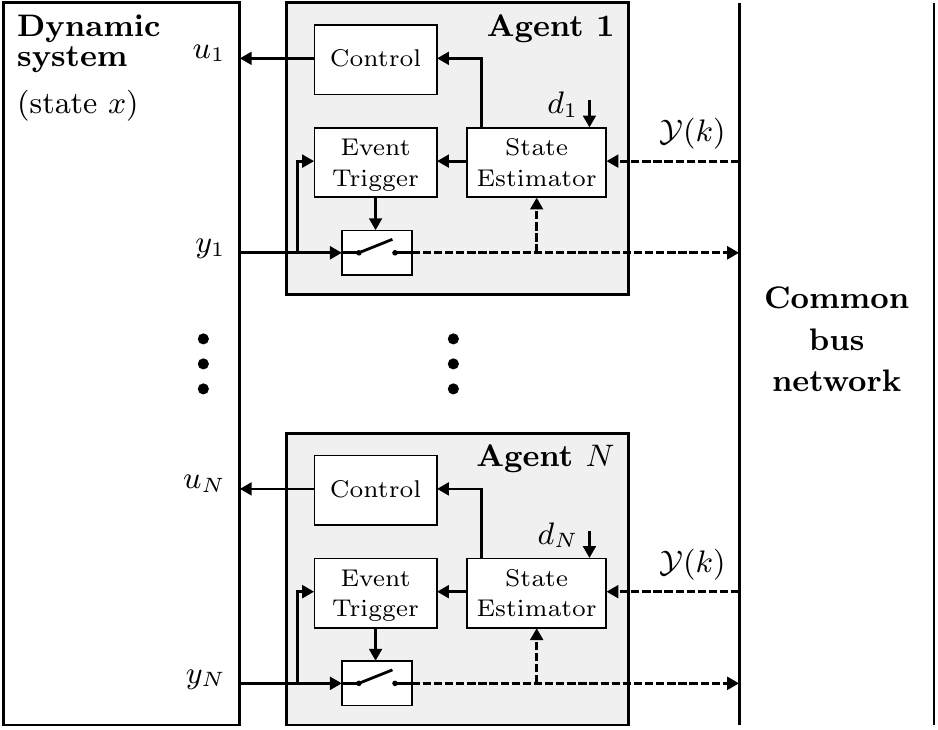}
\caption{Distributed, event-based control architecture in the experiments of \sect \ref{sec:experiments}.  $N$ spatially distributed agents observe and control a dynamic system and exchange data with each other via a common bus.  Compared to \fig \ref{fig:EBSE_MultiAgent}, the agents combine the functions of sensing/triggering, estimation, and additionally feedback control.  The State Estimator serves both for making the triggering decision and for feedback control.}
\label{fig:EBC_MultiAgent}  
\end{figure}

In order to study the effect of imperfect communication, we simulate random packet drops
such that a transmitted  measurement $y_i(k)$ is lost with probability 0.05, independent of previous drops.
Packet drops can be represented by the disturbance $d_i(k)$ in \eqref{eq:EBSE2_multi} as 
follows: if $y_\ell(k)$, $\ell \in I(k)$ is a measurement not received at agent $i$, then $d_i(k) = -L_\ell ( y_\ell(k) - C_\ell \hat{x}_i(k|k-1) )$ accounts for the lost packet. 
For simplicity, we assume that communicated inputs are never lost.

\subsubsection{Event-based design}
Each agent implements the event triggers \eqref{eq:eventTrigger_MB_multi} and \eqref{eq:eventTriggerCtrl}, the estimator \eqref{eq:EBSE1_multi}, \eqref{eq:EBSE2ideal_multi}, and the distributed control \eqref{eq:stateFeedbackDistributed}. Triggering decisions are made individually for the two sensors of each agent, but jointly for both inputs (\cf \tab \ref{tab:AgentsThermofluid}).


For the design of the centralized observer \eqref{eq:FCSE1}, \eqref{eq:FCSE2}, we chose $L =$ $\diag(0.1, 0.05, 0.1, 0.05)$ as observer gain, leading to stable $(I-LC)A$.
%
For this design, the inter-agent error dynamics \eqref{eq:epsij_dyn} are also stable:  by direct calculation, one can verify that \eqref{eq:LMI_cond} is satisfied with $P = \diag(500, 1, 500, 1)$ for all subsets $J \subseteq \{1,2,3,4\}$ (\cf supplementary material).  Lemma \ref{lem:bounded_eij} thus guarantees that \eqref{eq:epsij_dyn} is stable, and synchronous resetting \eqref{eq:syncAvg} not necessary. 

%

%

The state-feedback gain $F$ is obtained from an LQR design, 
 which involves full couplings between all states in contrast to the decentralized design in \cite{SiStGrLu15}.
The triggering thresholds 
are set to $\deltae_{11} = \deltae_{21} = 0.01 \, \text{m}$, $\deltae_{12} = \deltae_{22} = 0.2 \, \text{K}$, and $\deltac_1 = \deltac_2 = 0.02$.  

\subsubsection{Simulation results}
The state trajectories of a $2000 \, \text{s}$ simulation run under event-based communication are shown in \fig \ref{fig:exampleSimRum_thermoFluid_x}. 
Step-wise disturbances $v$
(gray shaded areas) with comparable magnitudes as in \cite{SiStGrLu15} cause the states to deviate from zero.
Especially at times when disturbances are active, the event-based estimate is slightly inferior to the centralized one, as is expected due to the reduced number of measurements.  
%
\begin{figure}[tbp]
\centering
\includegraphics{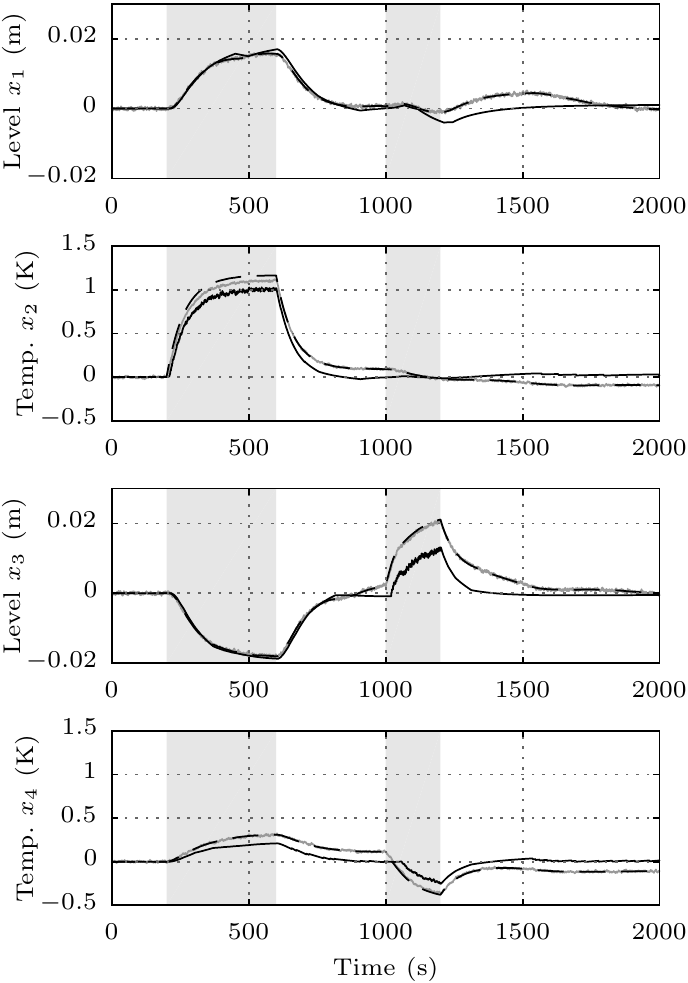}
\caption{State trajectories for the thermo-fluid simulation example.  {\bf Black (dashed):} the actual states $x$; {\bf black (solid):} event-based estimate $\hat{x}_1$ by Agent 1; {\bf gray:} centralized estimate $\hat{x}_\text{c}$.  The centralized estimate is shown for comparison and not available on any of the agents. The gray shaded areas indicate periods where step-wise process disturbances are active.}
\label{fig:exampleSimRum_thermoFluid_x}
\end{figure}

The average communication rates for event-based input and sensor transmissions 
 are given in \fig \ref{fig:exampleSimRum_thermoFluid_comm}.  Clearly, communication rates increase in the periods where the disturbances are active, albeit not the same for all sensors and inputs.  At times when there is no disturbance, communication rates are very low.  
%
%
\begin{figure}[tbp]
\centering
\includegraphics{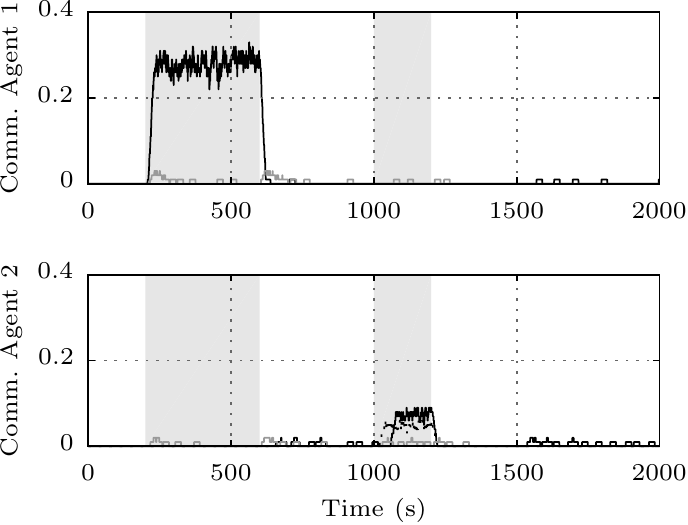}
\caption{Event-based communication rates for the thermo-fluid simulation example: communication of level measurements $y_{11}$ and $y_{21}$ in {\bf black (dotted)}, for temperature measurements $y_{12}$ and $y_{22}$ in {\bf black (solid)}, and for the inputs $u_1$ and $u_2$ in {\bf gray}.  Communication rates are computed as the moving average over 100 steps (0.0 meaning no communication and 1.0 full communication).}
\label{fig:exampleSimRum_thermoFluid_comm}
\end{figure}

Figure \ref{fig:exampleSimRum_thermoFluid_e12} shows the inter-agent error $e_{12}$.
Jumps in the error signals are caused by dropped packets, with decay afterward due to stable dynamics \eqref{eq:epsij_dyn}. 
%
\begin{figure}[tbp]
\centering
\includegraphics{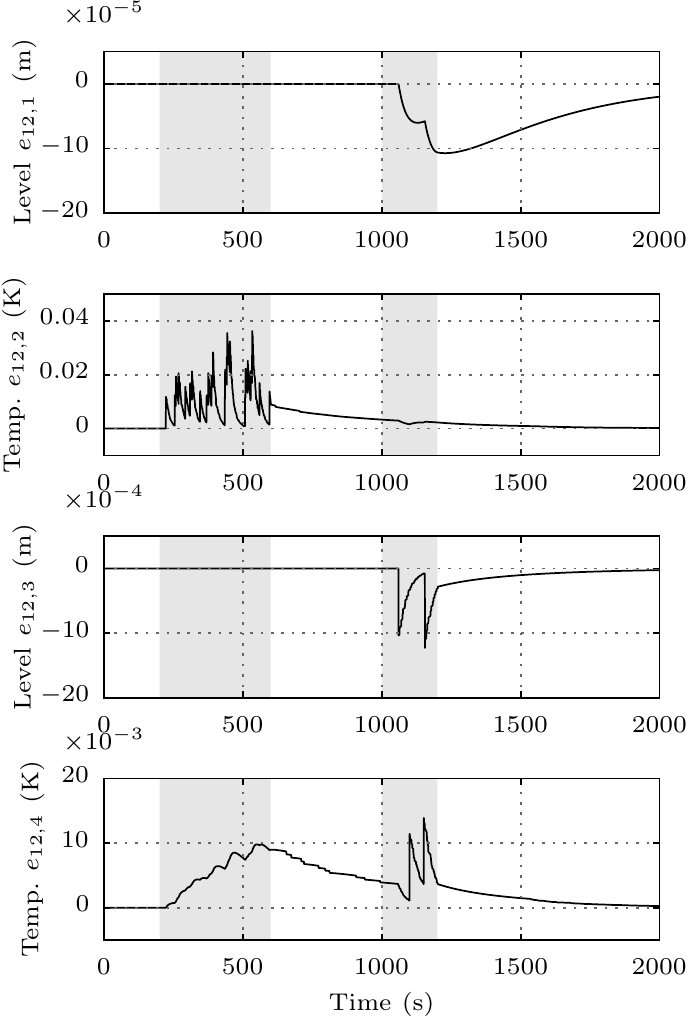}
\caption{Inter-agent error $e_{12}$ for the thermo-fluid example. Jumps in error are caused by packet drops, and the decay afterward is due to stable inter-agent dynamics \eqref{eq:epsij_dyn} as ensured by \Lem \ref{lem:bounded_eij}.}
\label{fig:exampleSimRum_thermoFluid_e12}
\end{figure}

\subsection{Experiments on the Balancing Cube}
\label{sec:experimentsBC}
The proposed emulation-based approach to event-based estimation was applied in \cite{Tr12} for stabilizing the Balancing Cube \cite{TrDAn12b} (see \fig \ref{fig:BC}).  
In this section, we summarize the main results from the experimental study reported in \cite{Tr12}.  For details, we refer to these citations.

\begin{figure}[tbp]
\centering
\includegraphics[width=9cm]{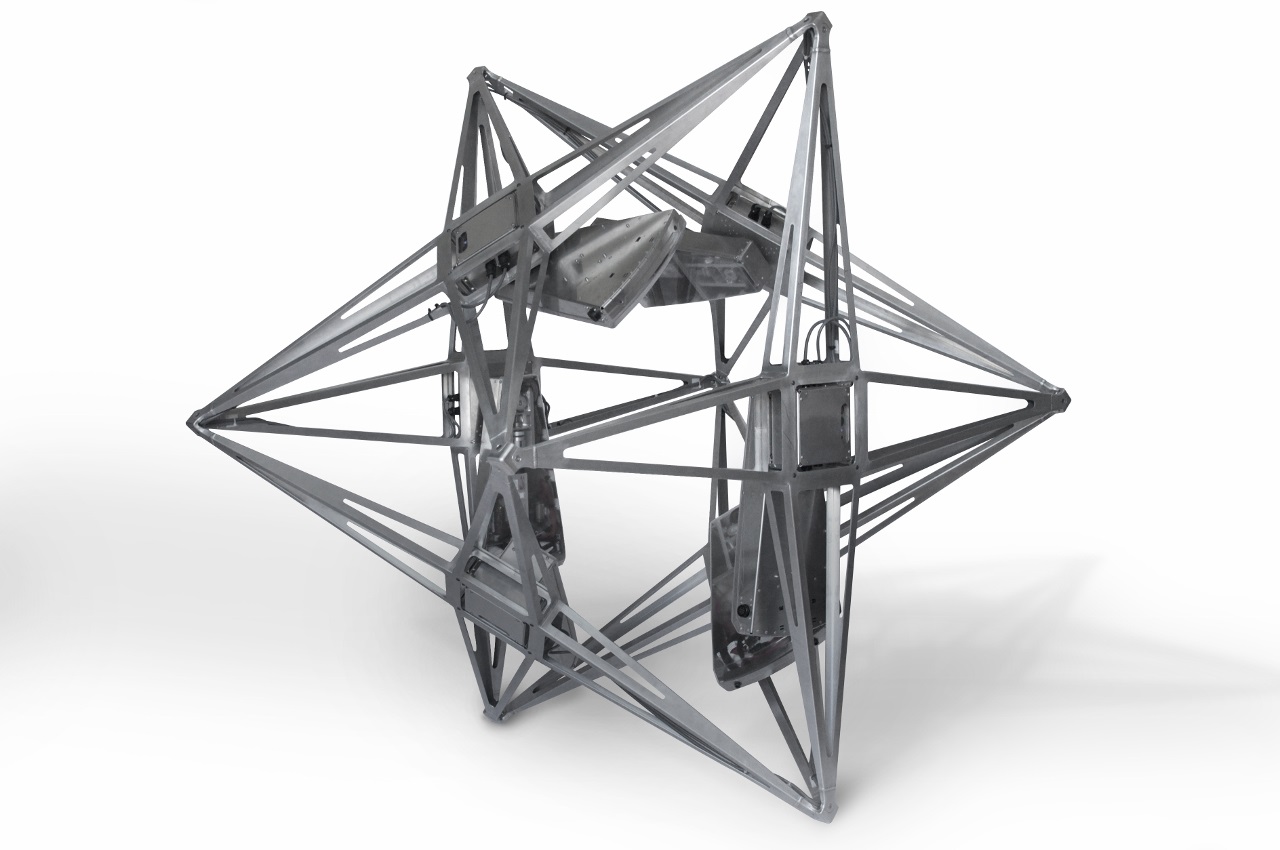}
\caption{The Balancing Cube \cite{TrDAn12b} can balance autonomously one any one of its corners or edges.  Six rotating arms, which constitute the control agents,  collaboratively keep the system in balance.  Here, the cube is shown while balancing on an edge, which was the configuration used for the experiments in \cite{Tr12}. (Photo: Carolina Flores, IDSC, ETH Zurich; \copyright 2012 IEEE. Reprinted, with permission, from \cite{TrDAn12b}.)
}
\label{fig:BC}
\end{figure}

\subsubsection{System description}
The cube is stabilized through six rotating arms on its rigid structure (see \fig \ref{fig:BC}).  Each arm constitutes a control agent equipped with sensors (angle encoder and rate gyroscopes), a DC motor, and a computer.  The computers are connected over a CAN bus, which supports the exchange of sensor data between all agents (including the worst case of all agents communicating within one sampling time $\Ts = 1/60 \, \text{s}$).  
Each agent thus combines the functions sensing, triggering, estimation, and control as in \fig \ref{fig:EBC_MultiAgent} ($N=6$).

\subsubsection{Event-based design}
A model \eqref{eq:system_x}, \eqref{eq:system_y} representing linearized dynamics about the equilibrium configuration shown in \fig \ref{fig:BC} is used for designing the centralized observer \eqref{eq:FCSE1}, \eqref{eq:FCSE2} (as a steady-state Kalman filter) and the controller \eqref{eq:stateFeedback} (linear quadratic regulator).
Each agent makes individual triggering decisions for its angle sensor and for its rate gyroscope with thresholds $\delta^\text{ang} = \unit[0.008]{rad}$ and $\delta^\text{gyro} = \unit[0.004]{rad/s}$, respectively.

In the experiments, control inputs $u_i$ were communicated periodically between all agents.  Synchronous resetting \eqref{eq:syncAvg} was \emph{not} applied, even though stability of the inter-agent error \eqref{eq:epsij_dyn} cannot be shown using \Lem \ref{lem:bounded_eij} because of unstable open-loop dynamics.  
Despite the absence of a formal proof, the system was found to be stable in balancing experiments.

\subsubsection{Experimental results}
Figure \ref{fig:BC_commRates} 
shows typical communication rates for some sensors during balancing.  
The desired behavior of event-based communication is well visible: feedback happens \emph{only when necessary} (\eg instability or disturbances).  Overall, the network traffic could be reduced by about 78\%
at only a mild decrease in estimation performance.

\begin{figure}[tbp]
\centering
\includegraphics{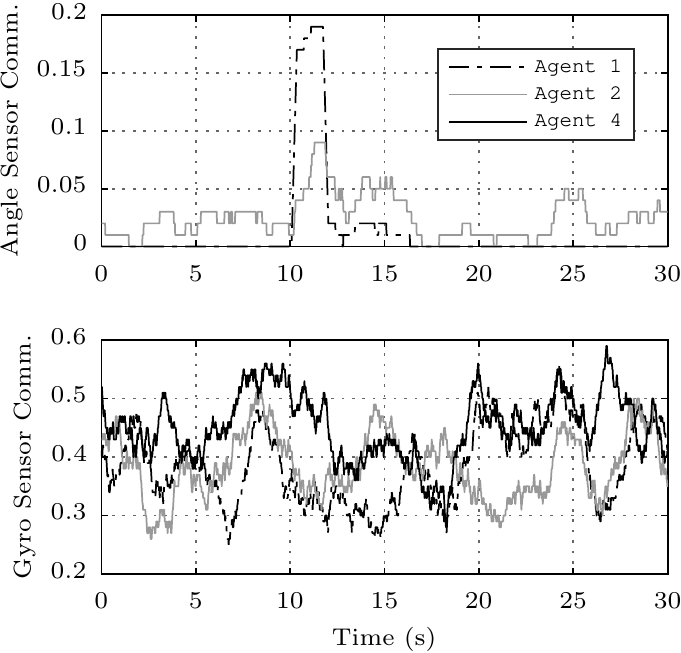}
\caption{Experimental communication rates on the Balancing Cube (reproduced from \cite{Tr12}). The communication rates (between 0 and 1) are computed as moving average over the last 100 time steps.  The rate gyroscopes generally transmit at higher rates than the angle sensors since they observe the unstable mode of the system.  The angle measurements can be predicted well from the process model; thus only little communication is necessary (\eg Agent 4 does not transmit over $30 \, \text{s}$). Caused by an external disturbance applied at $10 \, \text{s}$ on Agent 1 (pushing the arm), the communication rate of Agent 1 goes up temporarily.
}
\label{fig:BC_commRates}
\end{figure}

\section{Concluding remarks}
Simplicity of design and implementation are key features of the emulation-based approach to event-based state estimation developed herein.  
The approach directly builds on a classic centralized, linear, discrete-time state observer design.  Essentially, only the even triggers \eqref{eq:eventTrigger_MB_multi} and \eqref{eq:eventTriggerCtrl}, 
and (for some problems) synchronous resetting \eqref{eq:syncAvg} must be added.
The estimator structure, as well as the transmitted quantities 
remain unchanged, and no redesign of gains is necessary.  
The performance of the periodic design can be recovered by choosing small enough triggering thresholds, which simplifies tuning in practice.  Thus, implementation of the event-based system requires minimal extra effort,
and virtually no additional design knowledge.

With the proposed event-based method, the average communication load in a networked control system can be significantly reduced, as demonstrated in the simulations and experiment in this article.  
\section{Acknowledgment}
This work was supported in part by the Swiss National Science Foundation, the Max Planck Society, and the Max Planck ETH Center for Learning Systems.



\bibliographystyle{iet}        
{\small \bibliography{literature}}

\end{document}